\pdfoutput=1
\documentclass{article}
\usepackage{amssymb}
\usepackage{graphicx,amssymb,amsmath}
\usepackage[utf8]{inputenc}
\usepackage{authblk}
\usepackage[usenames, dvipsnames]{color}
\usepackage[utf8]{inputenc}
\usepackage[english]{babel}
\usepackage{amssymb}
\usepackage{comment}
\usepackage{amsthm}
\usepackage{caption}
\usepackage{algorithmic}
\usepackage{subfig}
\usepackage{enumitem}
\usepackage[export]{adjustbox}

\usepackage[margin=1truein]{geometry}
\usepackage[algoruled]{algorithm2e}
\SetKwInOut{Input}{input}
\SetKwInOut{Output}{output}
\SetKw{KwT}{True} 
\SetKw{KwF}{False}
\SetKw{KwOr}{or}
\SetKw{KwAnd}{and}
\DontPrintSemicolon
\newtheorem{theorem}{Theorem}
\newtheorem{lemma}{Lemma}
\newtheorem{corollary}[lemma]{Corollary}
\newtheorem{obs}[lemma]{Observation}

\graphicspath{{figures/}}
 
\newcommand{\etal} {{\it et al.}\xspace}

\definecolor {name} {rgb} {0.5,0.0,0.0}

\begin{document}


\title{Maximum-Area Triangle in a Convex Polygon, Revisited} 

\author[1]{Vahideh Keikha}
\affil[1]{\small Department of Mathematics and Computer Science, Amirkabir University of  Technology, Tehran, Iran}


\author[2]{Maarten L\"offler}
\affil[2]{\small Department of Information and Computing Sciences, Utrecht University, Utrecht, The Netherlands}
\author[1] {Ali Mohades}

\author[2]{J\'er\^ome Urhausen}
\author[2]{Ivor van der Hoog}

\maketitle

\begin{abstract}
We revisit the following problem: Given a convex polygon $P$, find the largest-area
inscribed triangle. We show by example that the linear-time algorithm presented in 1979 by Dobkin
and Snyder~\cite{45} for solving this problem fails. We then proceed to show that with a small
adaptation, their approach does lead to a quadratic-time algorithm. We also present a more
involved $O(n\log n)$ time divide-and-conquer algorithm. Also  we show by example that the algorithm  presented in 1979 by Dobkin
and Snyder~\cite{45} for finding the largest-area $k$-gon that is inscribed in a convex polygon fails to find the optimal solution for $k=4$. 
Finally, we discuss the implications
of our discoveries on the literature. 
\end{abstract}











\section{Introduction}

We revisit a classic problem in computational geometry: Given a convex polygon $P$, find the largest-area inscribed triangle. Figure~\ref {fig:example} illustrates the problem.
In 1979, Dobkin and Snyder~\cite{45} presented a linear-time algorithm to solve this problem. 
In this note, we present an example of a polygon on which their algorithm fails.
There exists, however, another linear-time algorithm for computing the largest inscribed triangle is presented by Chandran and Mount~\cite{chandran}, originally intended to solve the parallel version of the problem. Since the initial posting of this manuscript on arXiv~\cite {kluv}, two new linear-time algorithms for solving the problem have already been claimed by Kallus~\cite{kallus} and Jin~\cite{jin}.

The counter example is shown in Figure~\ref {fig:counter} and requires careful placement of the vertices of the polygon: our coordinates are integers in the range $[0, 5000]$, and a range of this order of magnitude seems to be necessary. We remark that  in~\cite{parvu} the authors implemented the presented algorithm in~\cite{chandran}, however, for the correctness of the implication, they produced a set of $10,000$ random convex polygons with range of vertices in $[0,1000]$, although our counter example clarifies that  the presented algorithm by Dobkin and Snyder~\cite{45}  always works correctly on any random convex polygon in the range $[0,1000]$. 

Our counter example was found by solving a system of quadratic equations.
We carefully analyze the underlying geometry and give insight into the reason for the failure of the original algorithm, and use this insight to create a new $O(n \log n)$ algorithm to solve the problem.

The study of geometric containment problems was initiated by Michael Shamos~\cite{shb}, who considered the question of finding the longest line segment in a convex polygon, also known in computational geometry as the {\em diameter} of the polygon.
Shamos presented a linear-time algorithm in his thesis~\cite {sham}, based on a technique which is now known under the name {\em rotating calipers}~\cite {rotcal}.
Dobkin and Snyder~\cite{45} also claimed a linear-time algorithm for computing the diameter, which was later found to be incorrect by Avis \etal~\cite{mmud}.

Dobkin and Snyder~\cite{45} originally claimed to have a generic linear-time algorithm for finding the largest inscribed $k$-gon inside a convex polygon (for constant $k$), but this claim was later retracted:
Boyce \etal~\cite {48} observe that the algorithm by Dobkin and Snyder fails for $k=5$ and instead present a $O(kn \log n + n \log^2 n)$ algorithm for finding the largest $k$-gon.
Aggarwal \etal~\cite{msearch} improve their result to $O(kn+ n \log n)$ time by using a matrix search method.
However, both algorithms still rely on the correctness of the Dobkin and Snyder algorithm for triangles, which to this day remains uncontested.
As we show here, this algorithm is also incorrect, and by extension, so are the solutions by Boyce \etal and Aggarwal \etal. In Section~\ref {sec:applications}, we discuss the implications of our finding in more detail.

Furthermore, We also show by  example that the presented algorithm~\cite{45} also fails to find the optimal solution for $k=4$. Our counter example is a polygon on 16 vertices in the range [0, 26500].
Our counter example, combined with the work in  \cite {mmud} and \cite{48} would suggest that the problem of finding the largest-area quadrilateral in linear time is still open. Although Boyce \etal~\cite{48} claimed one is presented by Shamos~\cite{shb}, but the cited manuscript cannot be found on-line.

\begin {figure}
  \includegraphics{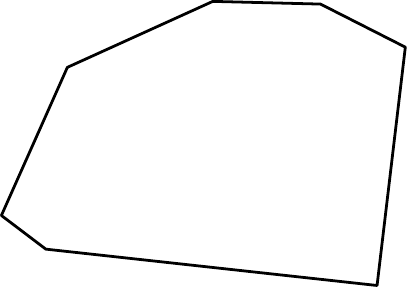}
  \centering
  \includegraphics{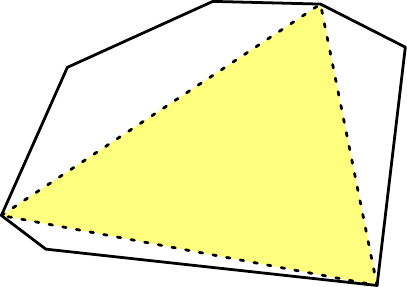}
  \caption {(left) A convex polygon. (right) The largest (by area) inscribed triangle.}
  \label {fig:example}
\end {figure}

\subsection{Related work}

The largest-area triangle problem falls in a broader class of geometric optimization problems, where the goal is to find some object inscribed inside another object. 
During the past 40 years, many optimization measures have been studied: 

Chang \etal ~\cite{ch2} studied the problem of computing the largest-area or largest-perimeter convex polygon that is inscribed in a simple polygon $P$. They presented an $O(n^7)$ time algorithm for the largest-area inscribed convex hull. They solved this   problem by using the fact that a largest-area convex polygon must be the intersection of $P$ and a set of $m$ half-planes defined by $m$ chords of $P$, where $m$ is smaller than the number of reflex vertices of $P$.  They also presented  an $O(n^6)$ time algorithm for the largest-perimeter  inscribed convex hull. Obviously, when the given polygon is convex,the solution to both problems is unique and can be found in linear time.
 Cabello \etal~\cite{cabello2014peeling} presented a  randomized near-linear-time $(1- \varepsilon)$-approximation algorithm for this problem, with running time  in $O(n(\log^2 n+(1/ \varepsilon^3)\log n+1/\varepsilon^4))$  time, and with probability at least $2/3$, the result has an area of at least $(1-\varepsilon)$ times the area of an optimal solution..

 Daniels \etal ~\cite{da3} presented an $O(n \log^2 n)$ time algorithm for finding  the largest-area axis-parallel rectangle that is inscribed in a simple polygon.  
They also presented  an $\Omega (n \log n)$ lower bound for this  problem. The running time also matches with the one of the best known algorithm for finding the largest-area axis-parallel rectangle that is inscribed in an  orthogonal polygons. 

 Cabello \etal ~\cite{linc2} studied the problem of finding the  largest-area or largest-perimeter  rectangle that is inscribed in a convex polygon. They presented an exact algorithm that runs in $O(n^3)$ time, and a $(1-\varepsilon)$-approximation algorithm that runs in $O(\varepsilon ^{-\frac{1}{2}} \log n+ \varepsilon ^ {-\frac{3}{2}})$ time; see also~\cite{alt1995,daniels1997,hall2006,knauer2012}.
 
  
DePano, Ke and O'Rourke~\cite{7} studied the problem of computing  the largest-area square and equilateral triangle contained in a convex polygon. 
  
   Jin and Matulel~\cite{8}  studied the problem of computing the largest-area parallelogram that is inscribed in a convex polygon $P$, and  presented an $O(n^2)$ time algorithm, which is based on the fact that largest-area parallelogram   must have all of its corners on the perimeter of $P$.

Melissaratos and Souvaine~\cite{melissaratos} studied the problem of computing the largest-area or perimeter triangle that is inscribed in  a non-convex polygon. They presented an $O(n^3)$ time algorithm for these problems.

When the input is an (unstructured) set of points in the plane, rather than a convex polygon, we may ask a similar question: what is the largest-area triangle or $k$-gon that uses only points of the given set as vertices? Clearly, we can attack this problem by first computing the convex hull of the given points. However, this takes $O(n \log n)$ time, and indeed Drysdale and Jaromczyk show that computing the largest-area or largest-perimeter $k$-gon must take at least $\Omega (n \log n)$ time for any $k \ge 2$, using a reduction from set disjointness~\cite {52}.


\subsection{Contribution}
  In this paper, we obtain the following results.
  
  \begin {itemize} [noitemsep]
    \item We present a $9$-vertex polygon on which the algorithm by Dobkin and Snyder for computing the largest-area triangle~\cite {45} fails. By extensions, the algorithm by Boyce \etal and Aggarwal~\etal for computing the largest-area $k$-gon~\cite {48,msearch} also fail.
    In particular, the example disproves Lemma 2.5 of Dobkin and Snyder~\cite{45} and Lemma 3.2 of Boyce \etal~\cite {48} (Section~\ref{sec:coex}).
    \item We analyze the geometry and present insight into the reason of the failure of these lemmas. We then use this insight to give a quadratic-time algorithm for computing the largest-area triangle in the same spirit as Dobkin and Snyder's algorithm, which we prove is correct (Section~\ref{sec:quadt}). 
    \item We then extend our analysis significantly and present a divide-and-conquer algorithm that works in $O(n \log n)$ time (Section~\ref{sec:dandc}).

  \item We present a $16$-vertex polygon on which the algorithm by Dobkin and Snyder for computing the largest-area quadrangle~\cite {45} fails.

  \end {itemize}

As an effect of the problem's central nature in computational geometry, a number of follow-up results that depend on largest-area triangles or $k$-gons, either directly by using the algorithm as a preprocessing step or indirectly by relying on false claimed properties, will have to be reevaluated~\cite {45,48,msearch,euro17,bhat}.
We discuss several of these in some detail in Section~\ref {sec:applications}.

\section{Preliminaries}

  In this section, we first review several core concepts introduced by Dobkin and Snyder~\cite {45} and Boyce \etal~\cite {48}, and then describe the algorithm by Dobkin and Snyder in detail.
  Although we are primarily concerned with triangles, we introduce some concepts more generally for $k$-gons; as we discuss in Section~\ref {sec:applications}, the question of finding the largest-area $k$-gon inside a convex polygon is also impacted (and reopened) by our results.

\subsection {Definitions}

Let $P$ be a convex polygon with $n$ vertices.
We say a convex polygon $Q$ is \emph {$P$-aligned} if the vertices of $Q$ are a subset of the vertices of $P$.
Note that there always exists a $P$-aligned largest-area $k$-gon inscribed in $P$ (assuming $k \le n$).
Boyce \emph{et al.} \cite{48} define a \emph{rooted} polygon $Q$ with root $r \in P$ to be any $P$-aligned polygon that includes $r$.  Let $P$ be a convex polygon. Two $P$-aligning polygons  are said to \textit {interleave}, if between every two successive
vertices of one, there is a vertex of the other (possibly coinciding with one of them) \cite{48}.\
Dobkin and Snyder \cite{45} define a \emph{stable} triangle to be a rooted triangle $T = pqr$, where $r$ is the root, such that any other $P$-aligned triangle $\triangle p'qr$ or $\triangle pq'r$ has area smaller   than (or equal to) $T$.
Henceforth, we will refer to such triangles as \emph {2-stable}, and to $p$ and $q$ as stable vertices.
We also define a \emph {3-stable} triangle to be a (rooted or unrooted) $P$-aligned triangle $T = pqr$ such that any other $P$-aligned triangle $\triangle p'qr$ or $\triangle pq'r$ or $\triangle pqr'$ has area smaller than (or equal to) $T$.
Note that in degenerate cases, there could be multiple (stable) triangles with equal area. 
In the remainder, we denote by $\Lambda$ the largest-area triangle, or, if it is not unique, any triangle with maximum area.
Note that $\Lambda$ is 3-stable.

\subsection {Dobkin and Snyder's triangle algorithm}
 We will now recall the \textit{triangle algorithm} \cite{45}, outlined in Algorithm~\ref {alg:triangle} and illustrated in Figure~\ref {fig:trialg}\footnote{For Simplification  of the implication, our algorithm is presented in C++ pseudo-code, it is easy to observe that it is  equivalent to the original presented  algorithm in ~\cite{45}. The same is also hold for $k=4$. }.

\begin {figure}
  \includegraphics{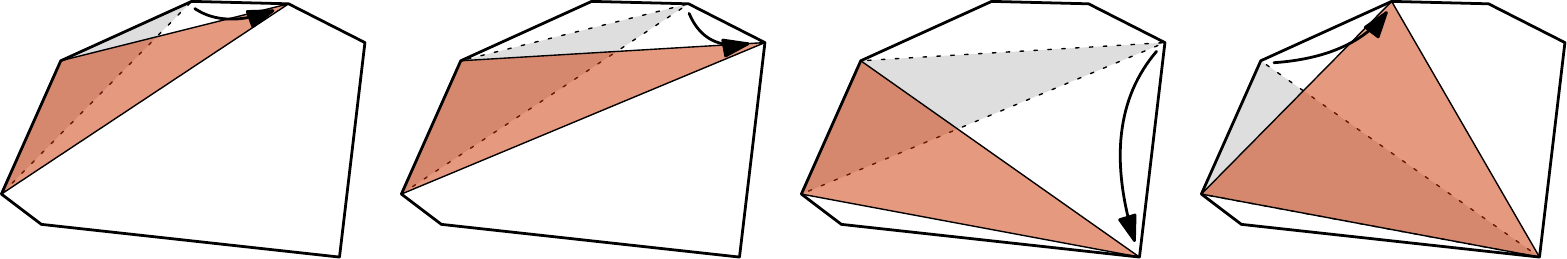}
   \centering
  \caption {The first four steps of Algorithm~\ref{alg:triangle}.}
  \label {fig:trialg}
\end {figure}

  Let $P=\{p_0,p_1,\ldots,p_{n-1}\}$. We assume that $P$ is given in a clockwise orientation. Assume an arbitrary vertex of $P$ is the root, assign this vertex and its two subsequent vertices  in the clockwise order on the boundary of $P$ to variables $a, b$ and $c$. 
  We then ``move $c$ forward'' along the  boundary of $P$ as long as this increases the area of $\triangle abc$.
  If we can no longer advance $c$, we advance $b$ if this increases the area of $\triangle abc$, then try again to advance $c$. If we cannot advance either $b$ or $c$ any further, we advance $a$.
  We keep track of the largest-area triangle found, and stop when $a$ returns to the starting position.
Since $a$ visits $n$ vertices and $b$ and $c$ each visit fewer than $2n$ vertices, the algorithm runs in $O(n)$ time (assuming we are given the cyclic ordering of the points on $P$).


Dobkin and Snyder \cite{45} claim that Algorithm~\ref {alg:triangle} computes the largest-area triangle inscribed in $P$. Their argument hinges on the following key lemma. 

\begin{lemma}[{\cite[Lemma 2.5]{45}}]
\label{doub}
Let $P=\{p_0, p_1, \ldots, p_{n-1}\}$ be a convex polygon. There exists an $i$
($0 \leq i \leq n-1$) such that the $p_i$-anchored maximum triangle is the largest-area triangle inscribed in $P$.
\end{lemma}

We claim that Lemma 2.5~\cite{45} 
is false.
In the lemma, the {\em $p_i$-anchored maximum triangle} refers to the largest triangle found by the algorithm while $a = p_i$. Note that this is not necessarily the same as the largest-area rooted triangle at $p_i$. This essential observation lies at the heart of the following construction.

\begin{algorithm} [H]
\SetAlgoLined

\caption{Triangle algorithm}


\label {alg:triangle}

     {\bf Input} {$P$: a convex polygon, $r$: a vertex of $P$}\\
     {\bf Output} {$T$: a triangle}\\
     {\bf Legend} Operation {\texttt{\textit{next}} means the next vertex in clockwise order of $P$}\\
     a = r\\
     b = \texttt{\textit{next}}(a)\\
     c = \texttt{\textit{next}}(b)\\
     m = $\triangle$abc \\
     
     \While{True}
{
       \While{$\triangle$ab\texttt{next}(c) $\geq $ $\triangle$abc or $\triangle$a\texttt{next}(b)c $\geq$ $\triangle$abc}
       {
	\If{$\triangle$ab\texttt{next}(c) $\geq $ $\triangle$abc}
       {
           c = \texttt{\textit{next}}(c)\;
       }
       \If{$\triangle$a\texttt{next}(b)c $\geq$ $\triangle$abc}
       {
         b = \texttt{\textit{next}}(b)\;
       }
}
       \If{$\triangle$abc $\geq$ $m$}
       {
         m = $\triangle$abc
       }

        a = \texttt{\textit{next}}(a)\;
       \If{a=r}
    {
       \Return m\;
     }
 }    
\end{algorithm}

\section{Counter-example to Algorithm~\ref {alg:triangle}}\label {sec:coex}

 In  Figure~\ref{fig:counter} we provide a polygon $P$ on $9$ vertices such that the largest-area inscribed triangle and the triangle computed by Algorithm~\ref {alg:triangle} are not the same.
 We use the following points:
  $a_1=(4752,4262), a_2=(3383,413), b_1=(759,2927), b_2=(4745,4322), c_1=(1213,691), c_2=(2506,4423), a_0=(3040,4460), b_0=(1000,1000), c_0=(5000,1000)$.
  The largest-area triangle is $\triangle a_0 b_0 c_0$; however, Algorithm~\ref {alg:triangle} reports triangle $\triangle c_0 c_1 c_2$ as the largest-area triangle.
 
\begin {figure}
  \includegraphics{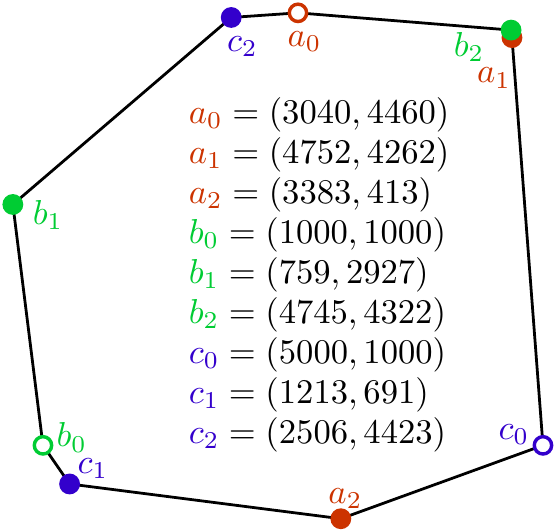}
   \centering
  \includegraphics{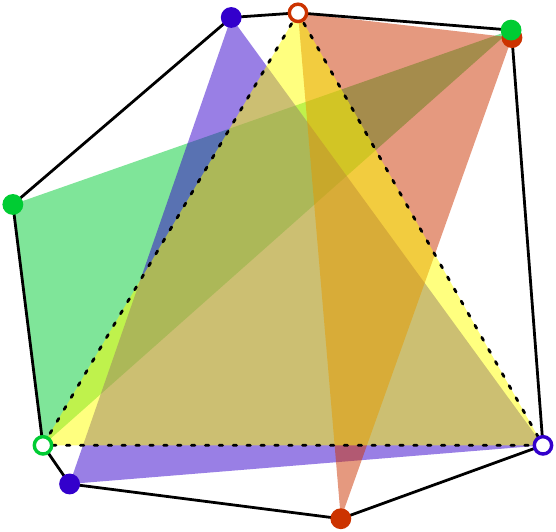}
  \caption {(left) A polygon on $9$ vertices. (right) Triangles $a_0 a_1 a_2$ (red), $b_0 b_1 b_2$ (green) and $c_0 c_1 c_2$ (blue) are all 2-stable, but smaller than triangle $a_0 b_0 c_0$ (yellow).}
  \label {fig:counter}
\end {figure}
 


We obtained $P$ by solving  the following problem. Let $\triangle a_0b_0c_0$ be the globally maximum area triangle on $P$, find six vertices $a_1$, $a_2$, $b_1$, $b_2$, $c_1$ and $c_2$ such that each of $\triangle a_1a_2a_0$, $\triangle b_1b_2b_0$ and $\triangle c_1c_2c_0$ are 2-stable triangles with area strictly less than $\triangle a_0b_0c_0$, and there is no other triangle on $P=\{a_0,b_0,c_0,a_1,a_2,b_1,b_2,c_1,c_2\}$ with area equal or larger than $a_0b_0c_0$. 
This resulted in a set of ${{9}\choose{3}} -1$ nonlinear constraints (many of which are redundant by Lemma \ref{lem:interl}) which delineate a very small but non-empty solution space. We were able to find an integer solution with 4-digit integers, but none with 3-digit integers.



\section{Some observations on the largest-area triangle}
 In this section, we argue that Algorithm~\ref {alg:triangle} does work on convex polygons in which there exists only one 2-stable triangle per vertex.
However, if there are multiple 2-stable triangles rooted at the same vertex, the algorithm only works if the first such triangle considered happens to be the largest one.
We investigate how many 2- and 3-stable triangles there can be, both per vertex and in total.

\subsection {The number of 2-stable triangles}


\begin{lemma}
	\label{lem:myinl}
	Let $r$ be any vertex of $P$.
	All 2-stable triangles rooted at $r$ are interleaving.
\end{lemma}

\begin {figure}
\includegraphics{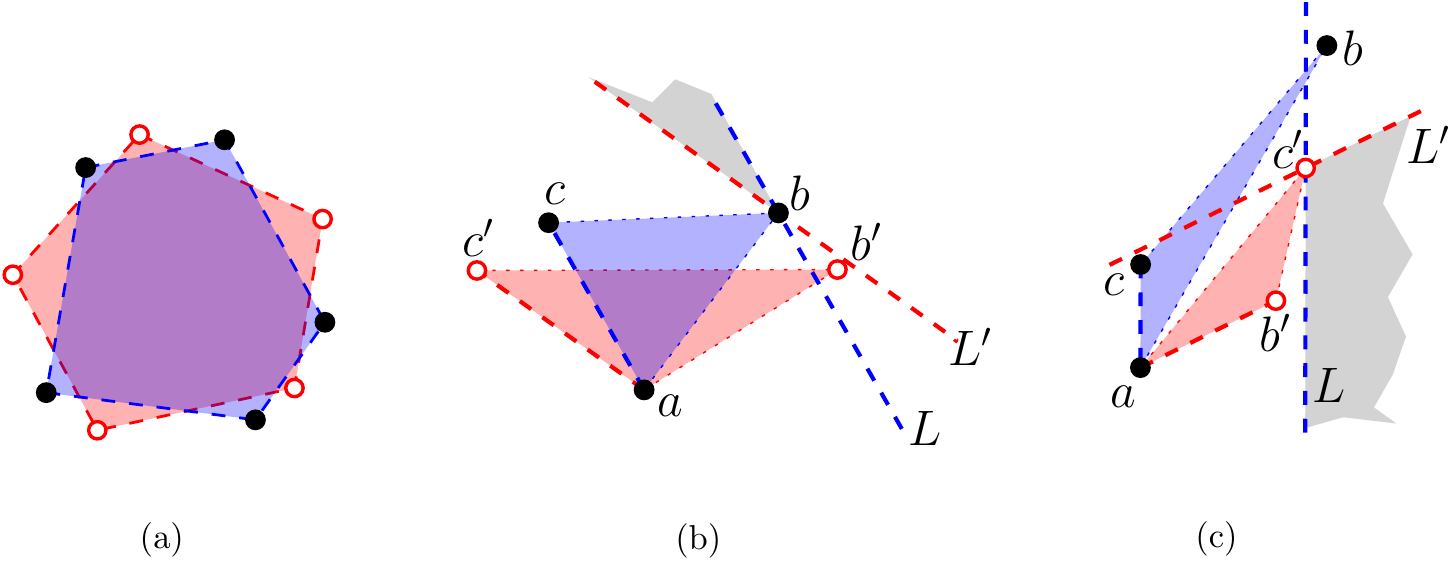}
\centering
\caption {(a) Interleaving polygons. (b,c) Stable triangles sharing a common root are interleaving. The gray areas can include vertex $b$. }
\label {fig:interleave-on-root}
\end {figure}

\begin{proof}
	Figure~\ref{fig:interleave-on-root}(b,c) illustrates the proof.
	Assume for the sake of contradiction that $P$ is a convex polygon and $a$ is a vertex on $P$ with several 2-stable triangles,
	such that there exists at least one 2-stable triangle $\triangle abc$ rooted at $a$ that is not interleaving with another 2-stable triangle $\triangle ab'c'$. W.l.o.g.\ both $\triangle abc$ and $\triangle ab'c'$ are ordered counterclockwise.
	
	First consider the case where $c$ occurs before $c'$ and $b'$ occurs before $b$ in counterclockwise order.
	Let $L$ and $L'$ be the lines through $b$ that are parallel to the lines $ac$ and $ac'$ respectively. 
	As $\triangle ab'c'$ is 2-stable, $b'$ has to be on the same side of $L$ than $a$. Additionally $b'$ has to be on the other side of $L'$ than $a$.
	It follows that $b'$ has to be placed after $b$ in counterclockwise ordering, contradicting the assumption.
	
	Second consider the case where $b'$ and $c'$ both occur before $b$ and $c$ in counterclockwise order.
	Let $L$ and $L'$ be the lines through $c'$ that are parallel to the lines $ac$ and $ab'$ respectively. As $\triangle abc$ is 2-stable, $b$ is on the other side of $L$ than $a$. As $\triangle ab'c'$ is 2-stable, $b$ is on the same side of $L'$ than $a$. This forces $b$ to be placed before $c'$ in counterclockwise ordering which again is a contradiction.
\end{proof} 

\begin{obs} \label {obs:linear}
	The number of 2-stable triangles rooted at any given vertex of a convex polygon is at most $ O(n)$.
\end{obs}
It is possible that for a given vertex $p_0$ on a convex polygon $P$, each edge $p_{0}p_i$ is an edge of a 2-stable triangle, as illustrated in Figure~\ref{fig:diff-manner}. 
Furthermore, Figure~\ref{fig:general-class} shows that any of these triangles could be the largest one; that is, the sequence of areas of the $2$-stable triangles rooted at $p_0$ is not necessarily increasing or decreasing. 
However, $p_0p_i$ can participate in at most two 2-stable triangles, namely, using the vertices that are farthest from the line through $p_0$ and $p_i$.
So, the number of 2-stable triangles rooted at any vertex of $P$ is $O(n)$. 

Furthermore, we can slightly alter the polygon in Figure~\ref{fig:general-class}(left) so that it becomes a polygon $P'$ with $O(n^2)$ $2$-stable triangles, see Figure~\ref{fig:general-class}(right). If we replace $p_0$ with $n$ new vertices $p_0, p_{0_1}, p_{0_2}, \ldots, p_{0_{n-1}}$, all on the boundary of $P$ (ordered clock-wisely) and close to each other, and such that the other points $p_i$ are far enough from each other, the resulting polygon $P'$ can have $O(n^2)$ 2-stable triangles.

Implicitly, Algorithm~\ref {alg:triangle} is based on the assumption that there is only a linear number of $2$-stable triangles that are comparable in size and to the size of the largest-area triangle.
We can alter the example in Figure~\ref {fig:general-class}(right) to make any of the $O(n^2)$ $2$-stable triangles 
the largest, but independent of such a change, Algorithm~\ref {alg:triangle} will always report the maximum area triangle out of the first found 2-stable triangles among all the vertices as the largest-area triangle.

Clearly, Observation~\ref {obs:linear} shows that the total number of 2-stable triangles is at most quadratic.

\begin{corollary}
	The total number of 2-stable triangles on a convex polygon is bounded by $O(n^2)$.
\end{corollary}

In~\cite{kluv} we show that correcting the Dobkin and Snyder algorithm to find all the 2-stable triangles results in a quadratic time algorithm.
\begin {figure}
\includegraphics{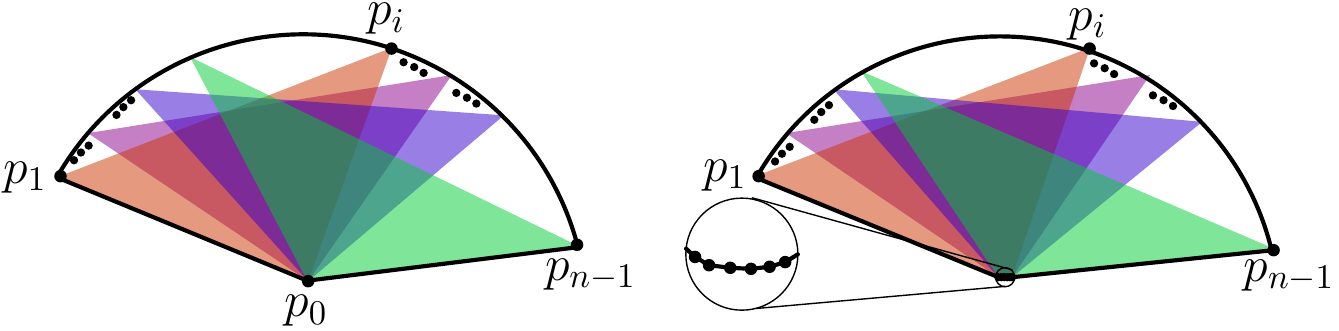}
\centering
\caption { General class of the polygon triangle algorithm~\cite{45} fails.}
\label {fig:general-class}
\end {figure}

\begin {figure}
\includegraphics{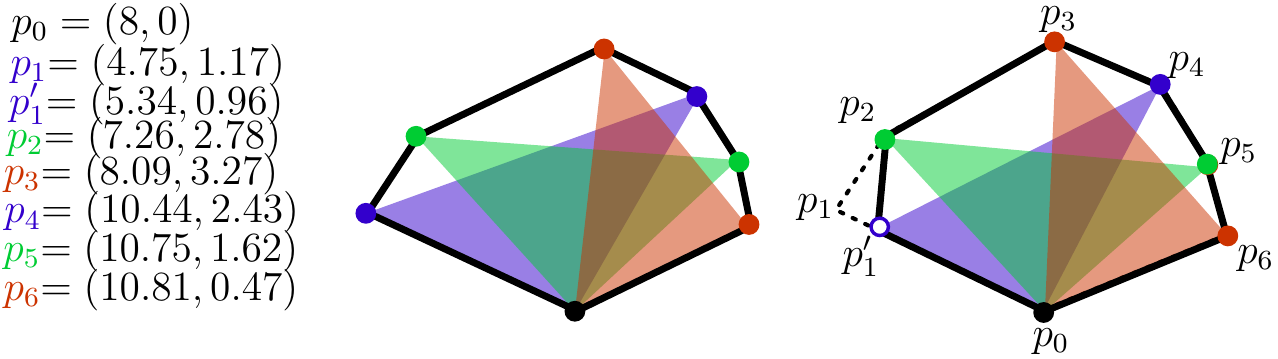}
\centering
\caption { A polygon with three $2$-stable triangles on a vertex; the blue triangle is the largest. If we move one vertex, there are still three $2$-stable triangles, but now the blue triangle is the smallest. 
}
\label {fig:diff-manner}
\end {figure}

\subsection {The number of 3-stable triangles}

\begin{lemma} Two 3-stable triangles on a convex polygon are always interleaving. \end{lemma}
\begin {figure}
\includegraphics{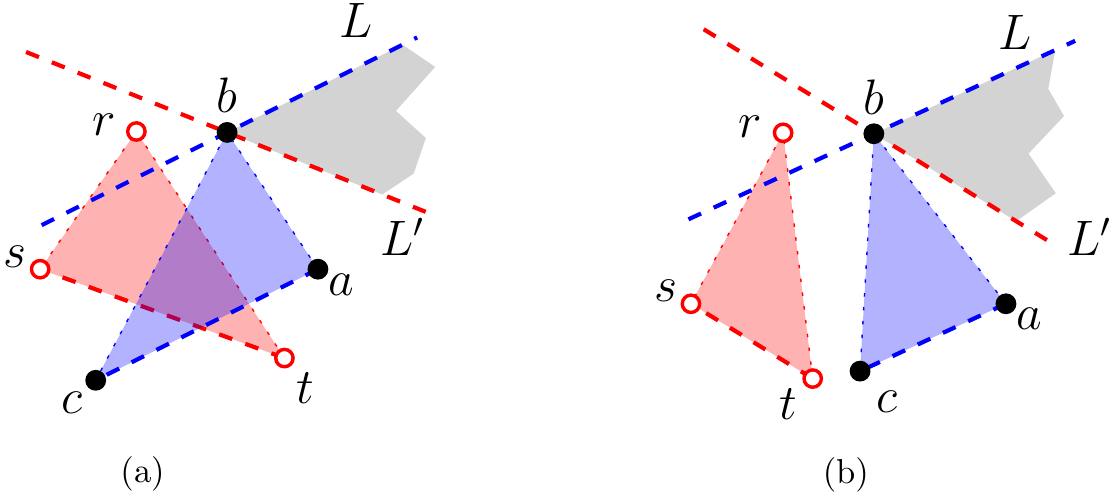}
\centering
\caption {The two cases we consider in order to prove that two 3-stable triangles interleave.}
\label {fig:Two_cases_stable_triangles_intersect}
\end {figure}
\begin{proof}
Let $\triangle abc$ and $\triangle rst$ be 3-stable triangles. Assume for the sake of contradiction that they are not interleaving. Following Lemma~\ref{lem:myinl} we know that the two triangles do not have a common vertex.
So either two vertices of $\triangle rst$ are between two vertices of $\triangle abc$ in counterclockwise order or the all three vertices of $\triangle rst$ are between two vertices of $\triangle abc$ in counterclockwise order.
W.l.o.g.\ we can assume that the counterclockwise order or the vertices is either $a, b, r, s, c, t$ or $a, b, r, s, t, c$.
The proof is the same for both cases and can be seen in Figure~\ref{fig:Two_cases_stable_triangles_intersect}. For the sake of simplicity we assume that the line $sa$ is horizontal and that the vertices $c$ and $t$ lie below that line.
Let $L$ and $L'$ be the lines through $b$ that are parallel to the lines $ac$ and $st$ respectively.
The vertex $r$ is closer to $ac$ than $b$ and $r$ is further away from $st$ than $b$. However due to $L$ having an upward slope and $L'$ having a downward slope, $r$ has to be before $b$ in cyclic ordering, which is a contradiction.
\end{proof}

\begin{lemma} The total number of 3-stable triangles on a given convex polygon is bounded by $O(n)$. \end{lemma}
\begin{proof}

Suppose $P=a_0,...,a_{n-1}$ is the convex  polygon. Let $P$ given in the clockwise ordering. Let $X_0$ be the number of 3-stable triangles  rooted at $a_0$, and let $x_{0,1}$ and $x_{0,n_{a_0}}$ denote, respectively, the first and last   3-stable triangles rooted at $a_0$, etc. 
Interleaving property of 3-stable triangles follows that if we  move  $a_0$ forward to $a_1$, the  second and third vertices of  triangle $x_{1,1}$  can only be located after or on the second and third vertices of triangle $x_{0,1}$. Also exactly one triangle rooted at $a_1$ can share both of the second and third vertices with one 3-stable triangle rooted at $a_0$, otherwise contradict the convexity of $P$. Note that the same argument is hold for all the successive vertices of $a_1$.

If there is a triangle $x_{1,j}$ with a common second vertex with $x_{0,i}$, then the third vertex of $x_{1,j}$  either coincides  with the third vertex of $x_{0,i}$, or moves forward on the clockwise ordering of $P$. Thus the range of movement of the second vertex of 3-stable triangles rooted at $a_1$ is bounded by the second vertices  of triangles $x_{0,1}$ and $x_{0,n_{a_0}}$ in the clockwise ordering, and the range of movement of the third vertex of 3-stable triangles rooted at $a_1$ is bounded by the third vertices  of triangles $x_{0,1}$ and $a_{n-1}$.
Let $i_0$ denote the index of the second vertex of $x_{0,n_{a_0}}$ and so on. 

Obviously we never passed through $a_{n-1}$. Also $i_j$ for $j=0,...,n-1$ are ordered by their indices. Also since the 3-stable triangles are interleaving, the range of the indices we trace by the third vertex of any 3-stable triangles do not overlap (or overlap with at most one vertex). All together   $\sum_{j=0}^{{n-1}} (i_{j+1}-i_j)$ is bounded by $O(n)$, and thus the total number of 3-stable triangles is bounded by $O(n)$.

\end{proof}


\section{A quadratic-time triangle algorithm} \label {sec:quadt}
In this section, we present a quadratic-time algorithm to find the largest-area inscribed triangle. Let $P=\{p_0, p_1, \ldots, p_{n-1}\}$ be a given convex polygon.
Recall that the largest-area triangle is 3-stable.
The idea of the algorithm is to find all 3-stable triangles in $P$: for each vertex $p_i$ we find all 2-stable triangles rooted at $p_i$ in a single linear pass; because all 3-stable triangles are also 2-stable for some vertex, we find all 3-stable triangles.

In step $i$ of Algorithm~\ref {alg:trianglequad}, we let $a = p_i$ be the root, and start searching from $a$ and its two subsequent vertices $b$ and $c$ on $P$. In contrast to Algorithm~\ref {alg:triangle}, each time we move $a$, we reset $b$ and $c$, but just like in Algorithm~\ref {alg:triangle}, each time we move $b$, $c$ stays where it is. 
This means each step of the algorithm now takes linear time, and the total algorithm takes quadratic time.

\subsection {Correctness}

We will start the correctness proof of Algorithm~\ref {alg:trianglequad} by the following lemma.

\begin{lemma} \label{quadl}
 Algorithm~\ref{alg:trianglequad}  considers all 2-stable triangles.
\end{lemma}
\begin{proof}
Suppose the lemma is false.
Then, there exists at least one 2-stable triangle $\triangle rst$ rooted at $r$ that the algorithm cannot find when $a=r$.
First assume that during the algorithm, at some point, $b$ will reach $s$. While $b$ is at $s$, $c$ will traverse some sequence $X$ of vertices of $P$; let $x$ be the first vertex of $X$. If $t$ is not in $X$, there are two cases: $t$ comes before $X$, or $t$ comes after $X$.

If $t$ comes before $X$, then $c$ already passed $t$ before $b$ reached $s$. This means that $b$ was at some point $u$ when $c$ passed $t$. But now, $\triangle rux$ and $\triangle rst$ would both be $2$-stable, but not interleaved, contradicting Lemma~\ref {lem:myinl} (see Figure~\ref{fig:all2stable}(left)).

If $t$ comes after $X$, then $b$ already moved away from $s$ before $c$ reaches $t$, say, when $c$ was at another vertex $v$. But then, $\triangle rsv$ was $2$-stable. However, $\triangle rst$ is also $2$-stable, contradicting the definition of $2$-stability.

Now suppose we missed $\triangle rst$ because $b$ did not reach $s$ before $c$ reaches $a-1$.
But then, $c$ passed both $s$ and $t$, so when this happens $\triangle rbc$ and $\triangle rst$ are not interleaved, a contradiction with  Lemma~\ref{lem:myinl} (see Figure~\ref{fig:all2stable}(right)).
\end{proof}

\begin {figure}
  \includegraphics{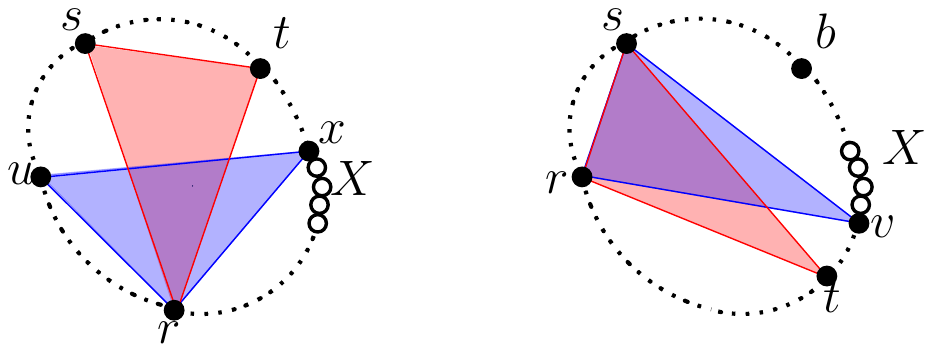}
   \centering
  \caption {Illustration of Lemma~\ref{quadl}. }
  \label {fig:all2stable}
\end {figure}

\begin{algorithm}[H] 
\caption{Quadratic-time triangle algorithm}
	\label {alg:trianglequad}
	{\bf Input} {$P={p_0,\ldots,p_n}$: a convex polygon, $p_0$: a vertex of $P$}\\
	{\bf Output} {$\Lambda$: Maximum-area triangle}\\
	a = $p_0$\\
	b = \texttt{\textit{next}}(a)\\
	c = \texttt{\textit{next}}(b)\\
	m = $\triangle$ abc\\
	\While{True}
	{
		
		\While{$c \ne a$}
		{
			\While{$\triangle$ab\texttt{next}(c) $\geq $ $\triangle$abc}
			{
				c = \texttt{\textit{next}}(c)\\
			}
			
			\If{$\triangle$abc $\geq $$\Lambda$}
		    {
			m = $\triangle$abc
		    }
				b = \texttt{\textit{next}}(b)\\
		}

			a = \texttt{\textit{next}}(a) \\
	\If{a=r}
	{
	\Return m\\
	} 
	b = \texttt{\textit{next}}(a) \\
	c = \texttt{\textit{next}}(b) \\
		
}
	
\end{algorithm}

\begin{theorem}
Algorithm~\ref{alg:trianglequad} will find the largest-area triangle in $O(n^2)$ time.
\end{theorem}

\begin{proof}
The correctness of the algorithm depends on three facts. First,  the largest-area triangle is always a 3-stable triangle;  second, in the above procedure we will find all 3-stable triangles; and third, the set of all 3-stable triangles is a subset of the set of all 2-stable triangles.
The correctness of the first and third facts are obvious. In Lemma~\ref{lem:myinl} we  proved we will consider all the 2-stable triangles. Thus we can conclude the  Algorithm~\ref{alg:trianglequad} works correctly.

\end{proof}

 


\section{A divide-and-conquer triangle algorithm} \label {sec:dandc}
In this section, we will provide a more efficient algorithm for finding the largest-area   inscribed triangle on a convex polygon.
We will use the following previously established lemmas.

\begin{lemma} [{\cite[Lemma 2.2]{48}}]
\label{lem:interl}
  A globally largest-area $k$-gon and a largest-area rooted $k$-gon interleave. 
\end{lemma}

\begin{lemma}
\label{lart}
The  largest-area rooted triangle can be found in linear time.
\end{lemma}

\begin{proof}
The largest area triangle rooted at an arbitrary fixed vertex  $a$ of $P$ can be found via one step of Algorithm~\ref{alg:trianglequad}. In the correctness proof of that algorithm we mentioned that we can find all the 2-stable triangles on any given  root $a$ in linear time. The largest-area rooted triangle can also be found in linear time.
\end{proof}

\subsection{Algorithm}
In the first step of the algorithm we  choose an arbitrary vertex $a$ on $P$ and compute  the largest-area triangle rooted at $a$. We call this triangle ${T^1}_a$.  ${T^1}_a$ decomposes the boundary of $P$ into three intervals that share their endpoints. 
Let $m$ be the median vertex on the largest of these intervals (largest in terms of complexity).
We then compute the largest-area rooted triangle   on $m$, ${T^1}_m$. We call ${T^1}_a$ and ${T^1}_m$ \textit{dividing triangles}. 
The vertices of ${T^1}_a$ and ${T^1}_m$ subdivide $P$ into six intervals that share their endpoints; Figures~\ref{fig:dandca} and \ref{fig:dandcb} show two possible configurations. 

Recall that, by Lemma~\ref{lem:interl}, the largest-area triangle $\Lambda$  must interleave with both ${T^1}_a$ and ${T^1}_m$.
This implies that, once we fix the interval that contains one vertex of $\Lambda$, the other two vertices are constrained to lie in two pairwise disjoint intervals.
Depending on the configuration, there could be either one or two sets of three compatible intervals; if there are two, they must interleave.
\begin {figure}
\centering
\subfloat[ Interleaving dividing triangles construct two subproblems.]{
\includegraphics{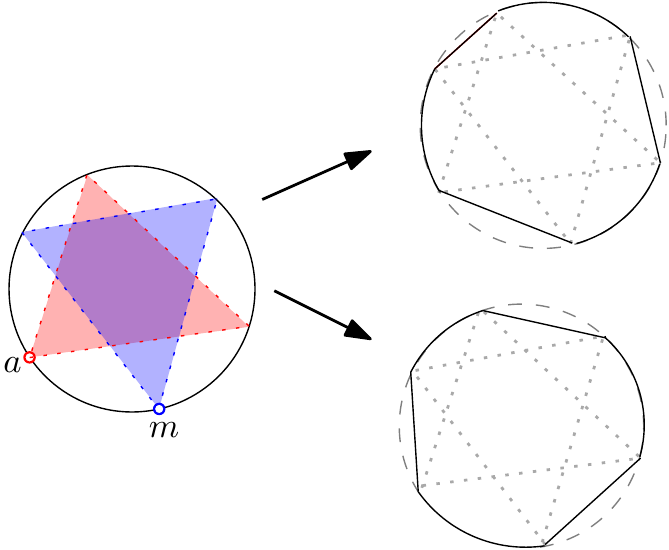}
 \label {fig:dandca}
}
\hspace{1cm}
\subfloat[ Non-interleaving dividing triangles construct one subproblem.]
{
\includegraphics{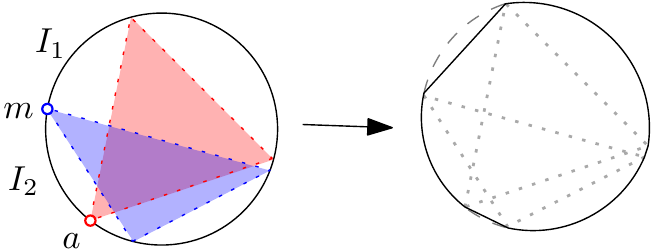}
 \label {fig:dandcb}
}

\caption {Dividing triangles and the resulting subproblems.}
\end {figure}
As a result, we construct either one or two smaller polygons $P'$ or $P'$ and $P''$ by directly connecting these intervals.

In the second step, we will repeat the above procedure again by finding another largest triangle $T_a^2$ rooted on an arbitrary vertex $a$ of $P'$ (or $P''$), and another largest triangle $T_m^2$ rooted on the median vertex $m$ of the largest sub-interval of $P'$ (or $P''$) induced by $T_a^2$.

Recursively repeating this, we show that for some subproblem $P^\star$ in step $i$ of the algorithm (note that in step $i$ there may be up to $2^{i}$ separate subproblems) if ${T^i}_a$ and ${T^i}_m$ are interleaving triangles in $P^\star$, we decompose the problem into two smaller subproblems, and if they are not interleaving, we get a single smaller subproblem.
In all cases, the size of each subproblem is between $\frac16$ and $\frac56$ times the size of the previous subproblem, and in the case where we have two subproblems, the sum of their sizes equals the size of the previous subproblem plus $6$.

We will repeat the procedure of constructing dividing triangles in each step on one or two smaller polygons, until our subproblems become triangles themselves; in this case we simply return the area of the triangle.
This procedure is outlined in Algorithm~\ref{alg:dandc}.

\begin{algorithm} [H] 


\caption{Divide-and-Conquer triangle algorithm}
\label {alg:dandc}

     {\bf Procedure} {\sc Largest-Triangle($P$)} \\
     {\bf Input} {$P$: a convex polygon}\\
     {\bf Output} {The largest-area triangle in $P$}\\
      \If{$|P| = 3$}
      {
        \Return $P$\;
      }

    \Else
    {  
         $a$ = an arbitrary root on $P$  \\
         $T_a$ = largest-area triangle rooted at $a$ \\
          $m$ = median point on the largest interval on $P$ between two vertices of $T_a$ \\
         $T_m$ = largest-area triangle rooted at $m$ \\
         $P'$,$P''$=sub-polygons constructed by interleaving intervals using $T_a$ and $T_m$ \\
     \If{ $T_a$ and  $T_m$ are interleaving}
     {
         \Return max (\textsc {Largest-Triangle($P'$)}, \textsc {Largest-Triangle($P''$)}) \\
        } 
        \ElseIf{$P'$ can include the largest-area triangle}
       {

	\Return  \textsc {Largest-Triangle($P'$)} \\
       }
      \Else 
       {
	
	\Return  \textsc {Largest-Triangle($P''$)} \\
       }
    
}


\end{algorithm}

\subsection {Time complexity}\label{sec:time}

We start analyzing the time complexity of the algorithm with the following lemma.

\begin{lemma} \label {lem:56}
  Let $P$ be a convex polygon with $n$ vertices.
  The (one or two) subproblems induced by $P$ have size at most $\frac56 (n + 6)$.
\end{lemma}
 
\begin{proof}
 The dividing triangles $T_a$ and $T_m$ decompose the boundary of $P$ into six intervals. Let $I_1$ and $I_2$ be the two intervals incident to $m$.
 Only one of $I_1$ or $I_2$ can  include a vertex of the largest-area triangle, otherwise the largest-area triangle would no longer interleave  both $T_a$ and $T_m$.
 Consider the $n-6$ (between $n-6$ and $n-4$) vertices not part of $T_a$ or $T_m$.
 Because of the choice of $m$, both $I_1$ and $I_2$ contain at least a factor $\frac16$ of these vertices, so $\frac16(n-6)$ each (see Figure~\ref{fig:dandcb}).
 Since one subproblem does not contain $I_1$, and the other subproblem does not contain $I_2$, each subproblem has size at most $n - \frac16 (n - 6) = \frac56 (n + 6)$.
\end{proof}

Note that, if $P$ splits into two subproblems $P'$ and $P''$, then $|P'| + |P''| \le n+6$.

So the recursive equation of the divide-and-conquer algorithm is

$$T(n)=\max \{T(\alpha(n+6))+T((1-\alpha)(n+6))+O(n+6), T(\alpha(n+6)+O(n+6)\}$$
where $\frac16 \le \alpha \le \frac56$, by Lemma~\ref {lem:56}.

By using $T(3)=1$ and considering the maximum part of the above equation, and using the method of Akra and Bazzi~\cite{akra}, the  recursion can be written as 
$$T(m)= T(\alpha m)+T((1-\alpha)m)+m$$
 and as $m$ is in $O(n)$, $T(n)$ is bounded by  $O(n \log n)$.




\subsection{Correctness}

For the correctness proof of Algorithm~\ref{alg:dandc} we will show that  in each step  of the recursion we always transfer all three vertices of the largest-area triangle $\Lambda$ to the same subproblem (or to both subproblems, if and only if $\Lambda = T_a$ or $\Lambda = T_m$). 
\begin{lemma}
In each step of Algorithm~\ref{alg:dandc}, there is always at least one subproblem containing all three vertices of the largest-area triangle $\Lambda$.
\end{lemma}

\begin{proof}
Let $P$ be a convex polygon, and consider the dividing triangles $T_a$ and $T_m$.
If $\Lambda = T_a$ or $\Lambda = T_m$, clearly both subproblems of $P$ contain all three vertices of $\Lambda$.

Otherwise, consider a subproblem $P'$ and suppose it does not contain all three vertices $p$, $q$, and $z$ of $\Lambda$; say (w.l.o.g.) $P'$ contains $p$ but not $z$.
We know $\Lambda = \triangle pqz$ must interleave both ${T_a}$ and ${T_m}$. But then $z$ must be in $P'$. 
This is only possible if $z$ is a vertex of $T_a$ or $T_m$. 
\end{proof}

\begin{theorem}
Algorithm~\ref{alg:dandc}  finds the largest-area triangle in $O(n \log n)$ time.
\end{theorem}

\section{Largest-area quadrangle}
Let $P$ be a convex polygon with $n$ vertices.
In the reminder, we denote by $\Lambda_{P,k}$  the largest (by area) $P$-aligned polygon with $k$ vertices. Also we denote $Q_{p,k}$ for $P$-aligned polygons with $k$ vertices. A polygon $Q$ is $k$-stable where it has  $k$ stable vertices.

 Note that all the vertices of $\Lambda_{P,k}$ are stable,  but a $k$-stable $Q_{P,k}$  does not necessarily coincides with the  $\Lambda_{P,k}$, as illustrated in Figure~\ref{fig:coin}.

Indeed the  idea of the presented method~\cite{45} was based on starting with a rooted $Q_{P,k}$ and moving the vertices of  $Q_{P,k}$ around the given polygon $P$ where keeping the cyclic ordering of $Q_{P,k}$ and increasing the area, and updating the area while finding a larger $k-1$-stable rooted polygon. 

This procedure will result in  keeping the sequence of the area of the potential solution only increasing. The authors~\cite{45} named this attribute as the \textit{unimodality of the area}, but we illustrated in Figure~\ref{fig:mmud} that keeping the unimodality will not result in finding the optimal solution necessarily.

\subsection {Dobkin and Snyder's  algorithm for $k=4$}
 We will now recall the \textit{quadrilateral  algorithm} \cite{45}, that is outlined in Algorithm~\ref {alg:quadrangle} and illustrated in Figure~\ref {fig:trialg}.

\begin {figure}
  \includegraphics{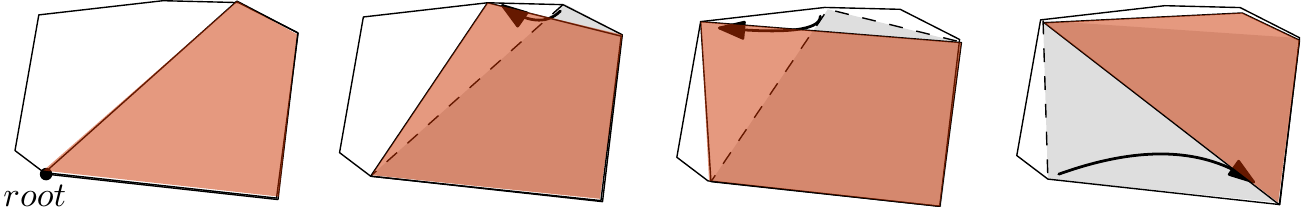}
   \centering
  \caption {The first three steps of Algorithm~\ref{alg:quadrangle}.}
  \label {fig:trialg}
\end {figure}

  Let $P=\{p_0,p_1,\ldots,p_{n-1}\}$. We assume that $P$ is given in a counter  clockwise orientation. Assume an arbitrary vertex of $P$ is is the root of the algorithm, assign this vertex and its three subsequent vertices  in the counter clockwise order on the boundary of $P$ to variables $a, b$, $c$ and $d$. 
  We then ``move $d$ forward'' along the  boundary of $P$ as long as this increases the area of $ abcd$.
  
If we can no longer advance $d$, we advance $c$ if this increases the area of $ abcd$, then try again to advance $d$.
If we can no longer advance $d$ and $c$, we advance $b$ if this increases the area of $ abcd$, then try again to advance $d$ and $c$.

 If we cannot advance either $d$, $c$ or $b$ any further, we advance $a$.
  We keep track of the largest-area quadrilateral found, and stop when $a$ returns to the starting position.
Since $a$ visits $n$ vertices and $d$, $c$ and $b$ each visit fewer than $2n$ vertices, the algorithm runs in $O(n)$ time (assuming we are given the cyclic ordering of the points on $P$).

Indeed, Algorithm~\ref{alg:quadrangle} is based on an observation that the largest inscribed quadrilateral treats as a unimodal function, which is not correct. 
 But, of course there is an  quadrilateral with some  vertices to be stable, that  Algorithm~\ref{alg:quadrangle} will find it in linear-time, but our counter example shows that the reported  quadrilateral does not necessarily equal to $\Lambda_{4,P} $. Furthermore, the reported quadrilateral is not even  4-stable.

\begin {figure}
  \includegraphics{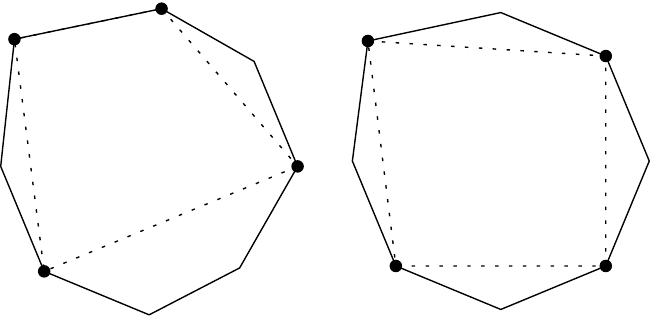}
   \centering
  \caption {Two 4-stable 4-gons inscribed in a convex polygon.}
  \label {fig:coin}
\end {figure}

\subsection{Counter-example to Algorithm~\ref {alg:quadrangle}}\label {sec:coex}

 In  Figure~\ref{fig:counter} we provide a polygon $P$ on $16$ vertices such that $\Lambda_{4,P}$  and  the largest-area inscribed quadrilateral computed by Algorithm~\ref {alg:quadrangle} are not the same. 
\begin {figure}
\includegraphics[scale=0.85]{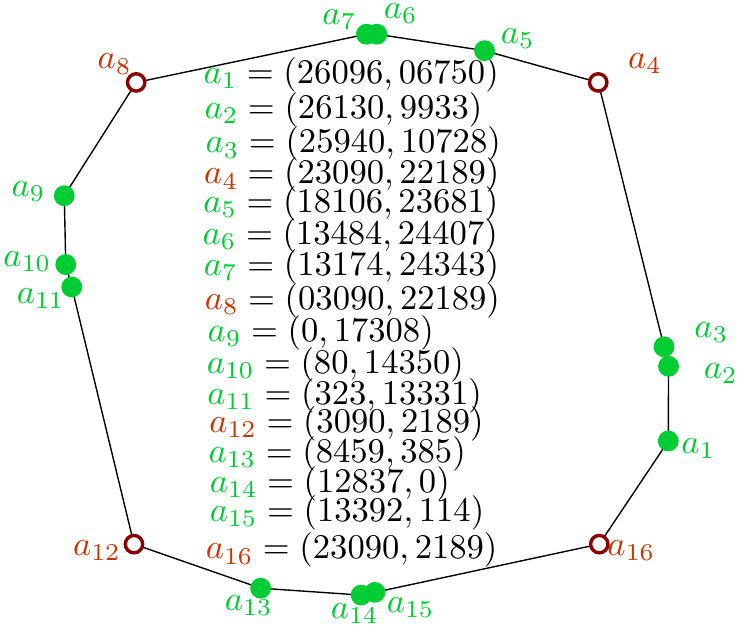}
\centering
\includegraphics[scale=0.85]{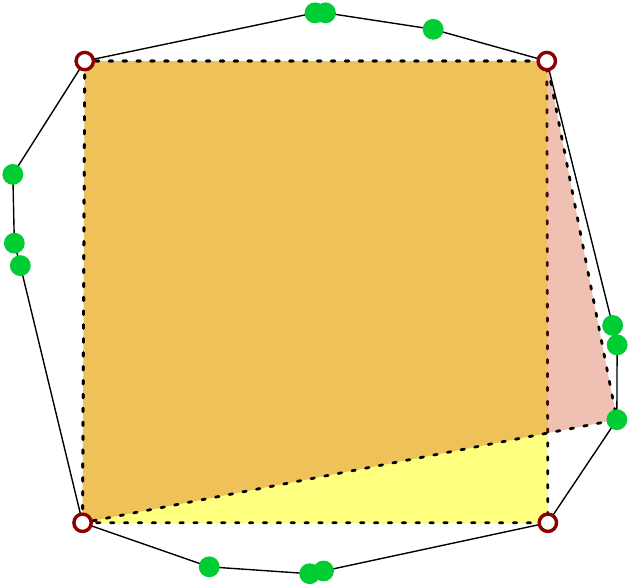}
\caption {(left) A polygon on $16$ vertices. (right)  The largest-area quadrilateral $a_4 a_8 a_{12} a_{16}$ (yellow), and the quadrilateral reported by Algorithm~\ref{alg:quadrangle}; $a_1 a_4 a_8 a_{12} $ (red).}
\label {fig:counter}
\end {figure}

\begin {figure}
\includegraphics{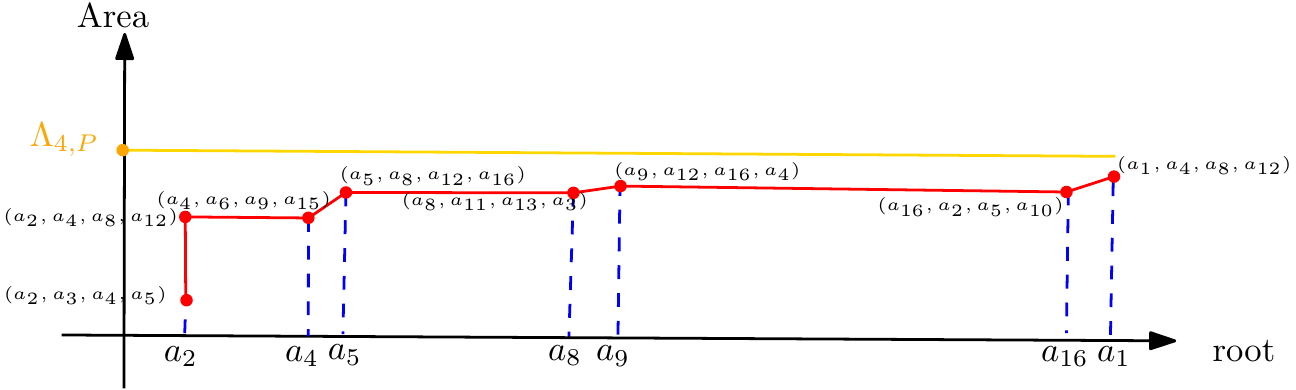}

\caption {Keeping the unimodality of the area of the potential solution during the algorithm will not result in the optimal solution necessarily.}
\label {fig:mmud}
\end {figure}
 We use the following points:
  $a_1=(26096,06750), a_2=(26130,9933), a_3=(25940,10728), a_4=(23090,22189), a_5=(18106,23681), a_6=(13484,24407),$ $ a_7=(13174,24343),  a_8=(3090,22189),  a_9=(0,17308),  a_{10}=(80,14350),  a_{11}=(323,13331),  a_{12}=(3090,2189),  a_{13}=(8459,385),  a_{14}=(12837 ,0),$ $ a_{15}=(13392,114) , a_{16}=(23090,2189)$. 
  The largest-area quadrilateral is $ a_4 a_8 a_{12} a_{16}$; however, Algorithm~\ref {alg:quadrangle} reports  $a_1  a_4 a_8 a_{12}$ as the largest-area quadrilateral, while  starting the algorithm from an arbitrary root. 
The results of running Algorithm~\ref{alg:quadrangle} while starting   on root $a_1$ are demonstrated on Figure~\ref{fig:mmud}.  
Thus,  the algorithm fails to find $\Lambda_{P,4}$ on any possible root.



\begin{algorithm} [H]
\SetAlgoLined

\caption{quadrilateral algorithm}


\label {alg:quadrangle}

     {\bf Input} {$P$: a convex polygon, $r$: a vertex of $P$}\\
     {\bf Output} { $m$: an quadrilateral}\\
     {\bf Legend} Operation {\texttt{\textit{next}} means the next vertex in counter-clockwise order of $P$}\\
     a = r\\
     b = \texttt{\textit{next}}(a)\\
     c = \texttt{\textit{next}}(b)\\
    d = \texttt{\textit{next}}(c)\\
     m = \texttt{\textit{area}}($abcd$) \\
     
     \While{True}
{
       \While{\texttt{\textit{area}}($abcd$) $\leq$ \texttt{\textit{area}}(abc \texttt{\textit{next}}(d))}
       {
		d = \texttt{\textit{next}}(d)\;

 \While{\texttt{\textit{area}}($abcd$) $\leq$ \texttt{\textit{area}}(ab \texttt{\textit{next}}(c)d)}
       {
		 c = \texttt{\textit{next}}(c)\;

  }   
      \While{\texttt{\textit{area}}($abcd$) $\leq$ \texttt{\textit{area}}(a \texttt{\textit{next}}(b)cd)}
       {
		b = \texttt{\textit{next}}(b)\;
      }

}
      
         m = max(\texttt{\textit{area}}($abcd$),m)

        a = \texttt{\textit{next}}(a)\;

       \If{a=r}
    {
       \Return m\;
     }

 \If{b=a}
    {
       b = \texttt{\textit{next}}(b)\;
	 \If{c=b}
		 {c = \texttt{\textit{next}}(c)\;
		 \If{d=c}
 			{d = \texttt{\textit{next}}(d)\;}}
     }

 }   
\end{algorithm}

\section {Implications} \label {sec:applications}

Our discovery directly or indirectly affects the results of the following studies.

\subsection {Largest-area $k$-gon inscribed in a convex polygon}
  As mentioned, the Dynamic Programming method presented by Boyce \etal~\cite{45} for finding the largest-area $k$-gon starts looking for the optimal answer  from the largest-area rooted triangle. Aggarwal \etal~\cite{msearch} improve their result to $O(kn+ n \log n)$ time by using a matrix search method. The method that  Boyce \etal~\cite{45} use to find the largest-area rooted triangle is again based on the assumption that there is only one stable triangle on each vertex of the polygon. So, they start  their algorithm from a non-optimal  answer. Also in their algorithm, they compute  $k$ intervals and they look for one point per interval, and these intervals are computed according to the starting situation. So they fail to find the largest area $k$-gon that is inscribed in a convex polygon.\par
In Lemma~\ref{lart} we proved that the globally largest-area rooted triangle can be found  in linear time. As such, it is relatively straightforward to correct their algorithm by changing the first step of their algorithm.

\subsection {Largest-Area Triangle inscribed in  a set of  imprecise points}
  Keikha \etal~\cite{euro17} consider the question of finding bounds on the area of the largest-area triangle in a set of {\em imprecise} points: points that are known to be in a given region in space.
  Their algorithm for computing the largest-area triangle on a given set of imprecise points modeled by parallel line segments was based on the Dobkin and Snyder algorithm~\cite{45}, and thus also fails to find the optimal answer. Two cases in particular are impacted.
  When the imprecise points are modeled as unit-length segments, their algorithm directly applies the largest-area triangle algorithm. This algorithm can easily be fixed by using our new divide-and-conquer algorithm instead, but this will result in a running time of $O(n \log n)$, while a running time of $O(n)$ was reported in~\cite {euro17}.
  In contrast, when the imprecise points are modeled as segments of arbitrary length, their solution does not directly apply the largest-triangle algorithm but is rather based on the stability of certain triangles. With a slightly extended analysis, the reported running time of $O(n^2)$ can still be achieved.

\subsection {Convex hull of a simple polygon}
 Bhattacharia \etal~\cite{bhat} applied the idea of unimodality of the vertical distance of a moving vertex on the boundary of a  convex polygon from one of its edges to compute the  convex hull of a simple polygon  $P$ in linear time. Specifically, they use the fact that there is one vertex $p_j$  on a convex polygon with maximum distance from an edge $p_1p_k$ of $P$. Using this fact, they decompose their problem into two subproblems by computing two half convex chains, one chain starting from $p_1$ and ending at $p_j$, and the other starting from $p_k$ and ending at $p_j$. On each subproblem, they use a stack $S$ which  stores some   vertices, such that the  triangles consisting of one vertex of $S$ and a fixed edge of $P$ are only increasing in the order in which they are stored in the stack. But as there is only one non-fixed vertex (for considering the area of a triangle with a fixed base) in each step of the algorithm~\cite{bhat}, the application of the idea of changing the area of an inscribed triangle by moving only one of its vertices~\cite{45} is still correct. Therefore, the correctness of the results of~\cite{bhat} is not impacted by our discovery.

\subsection {Critical Triangle}
Kallus~\cite{kal1} applied the Pentagon Lemma~\cite{45}  to compute the critical triangle $T$  in a convex compact subset of $\mathcal{R}^2$; $K$, where the critical triangle is an inscribed  triangle with maximum area  that the critical determinant of $K$  is equal to twice the area of $T$. 
 As the Pentagon Lemma~\cite{45} is still correct, the correctness of the results of~\cite{kal1} is not impacted by our discovery. They called $K$ is extensible if there is a domain $K'$ containing $K$ but different from it that has the same critical determinant as
$K$. Otherwise,  $K$ is inextensible.

\section {Discussion}

To summarize, we disproved the linear-time algorithm presented  by Dobkin and Snyder~\cite {45} for computing the largest-area inscribed triangle by presenting  a $9$-vertex polygon, on which the algorithm fails to output the optimal solution.
Dobkin and Schnyder's algorithm also fails to find the largest-area 4-gon in a convex polygon. For $k = 2$ and for $k \ge 5$, it was already known that their algorithm fails, so this now conclusively shows that the algorithm is wrong for all possible valies of k.

 We also  presented a divide-and-conquer algorithm for computing the largest inscribed $k$-gon which runs in $O(n \log n) $ time.  
Our findings reopen the question of whether the largest-area quadrangle in a convex polygon can be found in linear time. 

There remains a significant gap between the best know algorithm by Boyce \etal~\cite{48}, which runs in $O(kn + n \log n)$ time, and the lower bound of $\Omega(n \log n)$~\cite{msearch}. 

\section*{Acknowledgments}
The authors would like to thank everybody who has discussed this discovery and its implications with them over the past months, in particular Bahareh Banyassady, Ahmad Biniaz, Prosenjit Bose, Yoav Kallus, Kshitiz Kumar, Manish Kumar Bera,  Benjamin Raichel, Lena Schlipf, and the participants of Dagstughl Seminar 17171 on Computional Geometry, for their useful comments and insights.
Maarten L\"offler was partially supported by the Netherlands Organization for Scientific Research (NWO) under project no. 614.001.504.




\bibliographystyle{elsarticle-num}


\bibliography{sample}

\end{document}